\documentclass[showpacs,amsmath,amssymb,twocolumn,pra,superscriptaddress,floatfix]{revtex4-1}

\usepackage{graphicx} 
\usepackage{amsfonts}
\usepackage{amssymb}
\usepackage{amscd}
\usepackage{amsmath}    
\usepackage{enumerate}
\usepackage{epsfig}
\usepackage{subfigure}
\usepackage{xcolor}
\usepackage{amsthm}
\usepackage{physics}

\newtheorem{theorem}{Theorem}

\begin{document}
\preprint{APS/123-QED}
\title{Quantifying Coherence with Untrusted Devices}

\author{Xingjian Zhang}
\affiliation{Center for Quantum Information, Institute for Interdisciplinary Information Sciences, Tsinghua University, Beijing, 100084 China}

\author{Yunchao Liu}
\affiliation{Center for Quantum Information, Institute for Interdisciplinary Information Sciences, Tsinghua University, Beijing, 100084 China}

\author{Xiao Yuan}
\email{xiao.yuan.ph@gmail.com}
\affiliation{Department of Materials, University of Oxford, Parks Road, Oxford OX1 3PH, United Kingdom}


\begin{abstract}
{Device-independent (DI) tests} allow to witness and quantify the quantum feature of a system, such as entanglement, without trusting the implementation devices. Although DI test is a powerful tool in many quantum information tasks, it generally requires nonlocal settings. Fundamentally, the superposition property of quantum states, quantified by coherence measures, is a distinct feature to distinguish quantum mechanics from classical theories. In literature, witness and quantification of coherence with trusted devices have been well-studied. However, it remains open whether we can witness and quantify single party coherence with untrusted devices, as it is not clear whether the concept of DI tests exists without a nonlocal setting.
In this work, we study DI witness and quantification of coherence with untrusted devices. First, 
we prove a no-go theorem for a fully DI scenario, as well as a semi DI scenario employing a joint measurement with trusted ancillary states.
We then propose a general prepare-and-measure semi DI scheme for witnessing and quantifying the amount of coherence. We show how to quantify the relative entropy and the $l_1$ norm of single party coherence with analytical and numerical methods. As coherence is a fundamental resource for tasks such as quantum random number generation and quantum key distribution, we expect our result may shed light on designing new semi DI quantum cryptographic schemes.
\end{abstract}
\maketitle

\section{Introduction}\label{intro}
As a unique property in quantum information processing, {device-independent (DI) tests} allow the possibility of witnessing quantum properties with only observed statistics without relying on device implementations. The idea of device-independence first appeared in Bell tests~\cite{Bell}, where violations of Bell inequalities certify the existence of entanglement and hence rule out the possibility of any local hidden variable theory.  No assumptions on the implementation devices are made, and for this reason we say that Bell tests are fully DI. {From the observed statistics, we can even  determine certain unknown states (and uncharacterized measurements). This phenomenon is usually referred to as self-testing~\cite{mayers1998quantum,biparty2017self}, and Mayers and Yao first pointed out its great prospect in quantum cryptography.}
The independence of devices thus makes Bell inequalities a useful tool for many quantum information processing tasks such as entanglement witness~\cite{EW,Horodecki}, entanglement quantification~\cite{EQ}, quantum random number generation~\cite{OurQRNG,RMPQRNG}, and quantum key distribution~\cite{DIQKD,AcinDIQKD,NJP09,VidickDIQKD,MillerDIQKD,CaoDIQKD}.

Realizing a faithful violation of Bell inequalities puts very stringent requirements in practice~\cite{LoopholefreeViolation1,LoopholefreeViolation2,LoopholefreeViolation3}. The requirements for fully DI quantum information processing, such as DI quantum key distribution (QKD) and DI quantum random number generation (QRNG), are even higher~\cite{LF-DIQIP1,LF-DIQIP2,LF-DIQIP3,LF-DIQIP4,LF-DIQIP5,DIQRNG18}, as one needs to realize very high fidelity state preparation and high efficiency measurements.
Besides, not all entangled states can violate a Bell inequality~\cite{Werner,LHVPOV}. Therefore, the entangled states that violate no Bell inequalities cannot be witnessed in a DI manner.

Rather than distrusting all the devices, Buscemi~\cite{Buscemi} proposed a type of semi-quantum games. In conventional Bell tests, one needs to randomly choose classical inputs to determine the measurement bases. While in a semi-quantum game, the random classical inputs are replaced by general random quantum inputs.
Conventional Bell tests can be seen as a special case with orthogonal quantum input states. It is proved that any entangled state can be witnessed in at least one such game. Since one needs to trust the quantum input states, this scenario enjoys a measurement-device-independent (MDI) nature, and for this reason, we call it semi device-independent (semi DI). More generally, we define semi DI scenario such that it also includes cases where measurement devices are trusted while inputs are not, called source independent~\cite{SI}.
{Inspired by Buscemi's semi-quantum game, practical MDI
entanglement witness~\cite{MDIEW,MDISimulation,MEIEWexperiment1,MEIEWexperiment2,RRMDIEW} and MDI entanglement quantification~\cite{MDIEQ} have been proposed, reaching a balance between practicality and device-independence based on current experiment condition.}
Although MDIQKD~\cite{MDIQKD,MDIQKDEX,MDIQKD200,MDIQKD404}
was independently proposed before Buscemi's work, it can be unified
under semi DI framework as well.

{Most previous works about DI tests rely on nonlocal settings and focus on multi-partite quantumness, particularly entanglement. However, non-classicality can arise even in a single party quantum state. Quantum coherence, which describes the superposition of states on a given computational basis, is the most basic non-classicality that a single party can hold. }
Under a resource framework~\cite{Aberg,Baumgratz,RMPcoherence,RevModPhys.91.025001},
quantum coherence has been identified as a key resource in many quantum information processing tasks such as cryptography~\cite{COQKD}, quantum random number generation~\cite{Xiao,COQRNG}, and quantum thermodynamics~\cite{thermo1,thermo2,thermo3,thermo4,thermo5}.
Coherence is also closely related to other types of non-classicality such as discord and entanglement~\cite{CohwithEnt,CohDisc,nonclassicality,EntAndCoh,Unification}.
{Different types of coherence measures have been proposed~\cite{RMPcoherence}, and the problem of coherence witness and quantification with given measurements are also well studied in literature~\cite{ROC,ROA,CWEX}. Though DI tests and coherence are two well-studied fields in quantum information theory, it is not known whether we can perform (semi) DI tests in single party coherence witness and quantification. If the answer to this question is positive, then we can conclude that DI test is a general tool in quantum information processing, rather than a concept strongly linked with non-locality. This surely would shed light on designing new (semi) DI quantum cryptographic protocols relying on a single party state.}

In our work, we systematically study single party coherence witness and quantification with untrusted devices, including both fully and semi DI scenarios.
Specifically, after a brief review of fully and semi DI tests of entanglement and necessary concepts of quantum coherence in Sec.~\ref{Sec:background}, the paper is organized as follows.

(i) In Sec.~\ref{nogo}, we first show that it is impossible to witness single party coherence via a naive generalization of the existing fully and semi DI scenarios, that is, witnessing single party coherence either fully device-independently, or via a ``half part'' of Buscemi's semi-quantum game.

(ii) In Sec.~\ref{newsemi}, we propose a new semi DI scenario in which we can witness and quantify the coherence of an unknown quantum state. Instead of measuring an ancillary state and the unknown state jointly in Buscemi's semi-quantum game, only one state is prepared and measured for each run. We assume that the unknown state lies in the subspace spanned by the ancillary states. In this prepare-and-measure scenario, coherence quantification can be expressed as an optimization problem.

(iii) In Sec.~\ref{num} and~\ref{Cl1} we consider two kinds of coherence measure, the relative entropy and $l_1$ norm of coherence~\cite{Baumgratz}. Via numerical and analytical approaches, we prove it is possible to witness and quantify coherence with our new scenario. Advantages and limitations of different approaches are discussed with some examples, and with the analytical approach we explicitly analyze the validity of our scenario.

In this paper, we mainly focus our discussion on the qubit (two-dimensional) case, while we also show that our results can be naturally generalized to the qudit (high-dimensional) case.


\section{preliminary}\label{Sec:background}
\subsection{Fully and semi device-independent tests of entanglement}
Before discussing the problem of witnessing and quantifying single party coherence, we first briefly review how entanglement can be witnessed via fully or semi DI scenarios, as shown in Fig.~\ref{Fig:BellTests}. 
We express the scenarios in the form of nonlocal games for convenience and accordance.
Every Bell inequality is equivalent to a nonlocal game (including Buscemi's semi-quantum nonlocal games and the generalized Bell-like inequalities), and one can refer to~\cite{RMPnonlocal} for a better understanding on this matter.

\begin{figure}[htp]
\centering \resizebox{8cm}{!}{\includegraphics{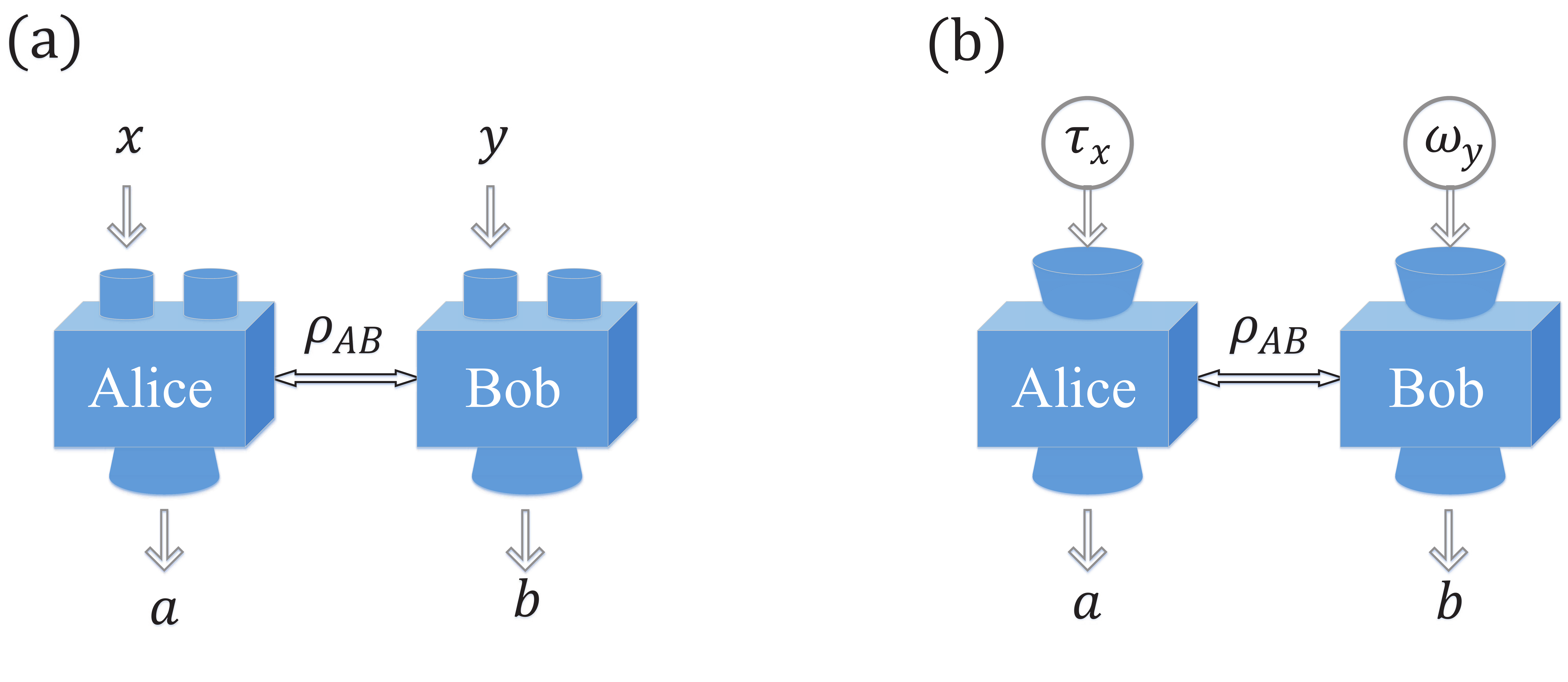}}
\caption{Witnessing entanglement via a fully DI and a semi DI scenario. (a) Fully DI scenario with classical inputs. Alice and Bob are given classical inputs, $x$ and $y$. Both of them perform a local measurement on their shared quantum state $\rho_{AB}$ and generate classical outputs, $a$ and $b$, respectively. The conditional probability distribution $p(a,b|x,y)$ is linked with their average score $S$. (b) Semi DI scenario with quantum inputs. The only difference in this scenario from the fully DI scenario is that quantum inputs, $\tau_x$ and $\omega_y$, rather than classical inputs, are given to Alice and Bob. The conditional probability $p(a,b|\tau_x,\omega_y)$ here means the probability of outputs $a$ and $b$, in the condition where inputs are $\tau_x$ and $\omega_y$. The probability distribution is then linked with their average score $T$.} \label{Fig:BellTests}
\end{figure}

In a bipartite Bell game, two players, say Alice and Bob, are given a few classical random inputs each, and are asked to generate some outputs independently. Spacelike separation is demanded such that once the game begins, no signal can be sent between the parties. In the quantum world, what Alice and Bob are capable of can be regarded as that they perform local measurements on a pre-shared quantum state. Then they are given a payoff depending on their outputs conditioned on the inputs {(the rule is known to both players)}. In a game with binary inputs and outputs, the average score Alice and Bob gain is
\begin{equation}\label{BellGame}
S=\sum_{a,b,x,y}\beta_{a,b}^{x,y}p(a,b|x,y),
\end{equation}
where $x$ and $y$ denote the inputs, $a$ and $b$ are the outputs generated by Alice and Bob, respectively. $\beta_{a,b}^{x,y}$ are scores corresponding to different circumstances. If we set $\beta$ equal to 1 or 0 in different circumstances with $\beta=1$ meaning that Alice and Bob win the game, then Eq.~\eqref{BellGame} is the winning probability. When possible inputs and outputs range from $\{0,1\}$ and Alice and Bob win iff $x \cdot y = a \oplus b$, Eq.~\eqref{BellGame} expresses a CHSH game~\cite{CHSH}. If $S>3/4$ is observed, we can conclude that Alice and Bob share an entangled state. Since no local hidden variable model can explain this phenomenon, we call such a state Bell nonlocal. In particular, if $S=\frac{2+\sqrt{2}}{4}$, we can imply that Alice and Bob share an EPR pair, in the sense of isometry.

As we have mentioned, Bell tests suffer from both theoretical and practical problems witnessing entanglement. {There is a gap between entanglement and Bell non-locality, hence there exist entangled quantum states that do not violate any Bell inequality \cite{RMPnonlocal}. A faithful violation of Bell inequalities requires that all experimental loopholes must be closed, yet loss and low efficiency of measurements can easily lead to a failure.} In his seminal work, Buscemi slightly modifies the conventional Bell nonlocal games~\cite{Buscemi}. The so-called semi-quantum game is all the same as a Bell game, except that general quantum inputs are allowed. We can write Bell-like inequalities
\begin{equation}\label{BuscemiGame}
T=\sum_{a,b,x,y}\beta_{a,b}^{x,y}p(a,b|\tau_x,\omega_y) \geq T_c,
\end{equation}
where $p(a,b|\tau_x,\omega_y)$ here represents the probability of outputs $(a,b)$ conditioned on quantum inputs p$(\tau_x,\omega_y)$, and $T_c$ is the maximum value Alice and Bob can achieve in a certain game with a separable state. It is proved that all entangled states can outperform separable states in at least one such semi-quantum game. This makes it possible to witness any entangled state via a semi DI approach. {Besides, it is now practical to prepare high-fidelity states, and with appropriate design, we can carry out loss-tolerant quantum information processing tasks based on this scenario, such as measurement-device-independent entanglement witness~\cite{MDIEW}.}

\subsection{Quantum coherence}
In this section, we review the resource theory of quantum coherence.
Consider a $d$-dimensional Hilbert space $\mathcal{H}_d$ and a computational basis $J = \{\ket{0},\ket{1},\dots,\ket{d-1}\}$, a state $\sigma$ is called incoherent if it only contains diagonal terms in its density matrix
\begin{equation}\label{inco}
  \sigma = \sum_{i=0}^{d-1}p_i\ket{i}\bra{i}.
\end{equation}
When a state $\rho$ cannot be written in this form, we call it coherent state and measure its coherence by adapting a proper coherence measure~\cite{Baumgratz}.

There are many functionals $C$ which can be used as coherence measures~\cite{RMPcoherence}. In this paper, we consider two distance-based quantifiers of coherence, the relative entropy of coherence and the $l_1$ norm of coherence
\begin{equation}\label{rel}
  C_{RE}(\rho) = S(\rho||\Delta(\rho)) = S(\Delta(\rho)) - S(\rho),
\end{equation}
\begin{equation}\label{l_1 norm}
  C_{l_1}(\rho) = \underset{i\neq j}\sum|\rho_{i,j}|,
\end{equation}
where $\Delta(\rho)=\sum_{i=j}\rho_{i,j}\ket{i}\bra{j}$, $S$ is the von Neumann entropy, and $S(\rho||\sigma)=\Tr[\rho \log_2\rho]-\Tr[\rho \log_2\sigma]$. The relative entropy of coherence has a clear physical interpretation, which is the distance between a state $\rho$ and the set of incoherent states. {This measure is related with intrinsic randomness against quantum adversary~\cite{yuan2016interplay,Liu2018quantum}, and quantifies the asymptotically distillable coherence under incoherent operations \cite{winter16}.} The $l_1$ norm coherence quantifier relates to the off-diagonal elements of the considered quantum state and is a widely used quantifier that intuitively shows the physical interpretation of coherence.

\section{Device-independent tests of coherence:  no-go for existing scenarios}\label{nogo}
Entanglement is a special form of coherence that only exists between at least two parties. Focusing only on a single party, it is tempting to ask whether there are counterparts of the (semi) DI scenarios in the problem of witnessing and quantifying single party coherence.
Different from (semi) DI tests of entanglement, it is a single party problem now, and only one untrusted device is involved essentially, as shown in Fig.~\ref{Fig:NOGO}. Surprisingly, neither of the above two methods can be directly generalised for witnessing coherence. We'll show that in either case, we can always find an incoherent state and some measurement to reconstruct a given probability distribution.

\begin{figure}[htp]
\centering \resizebox{8cm}{!}{\includegraphics{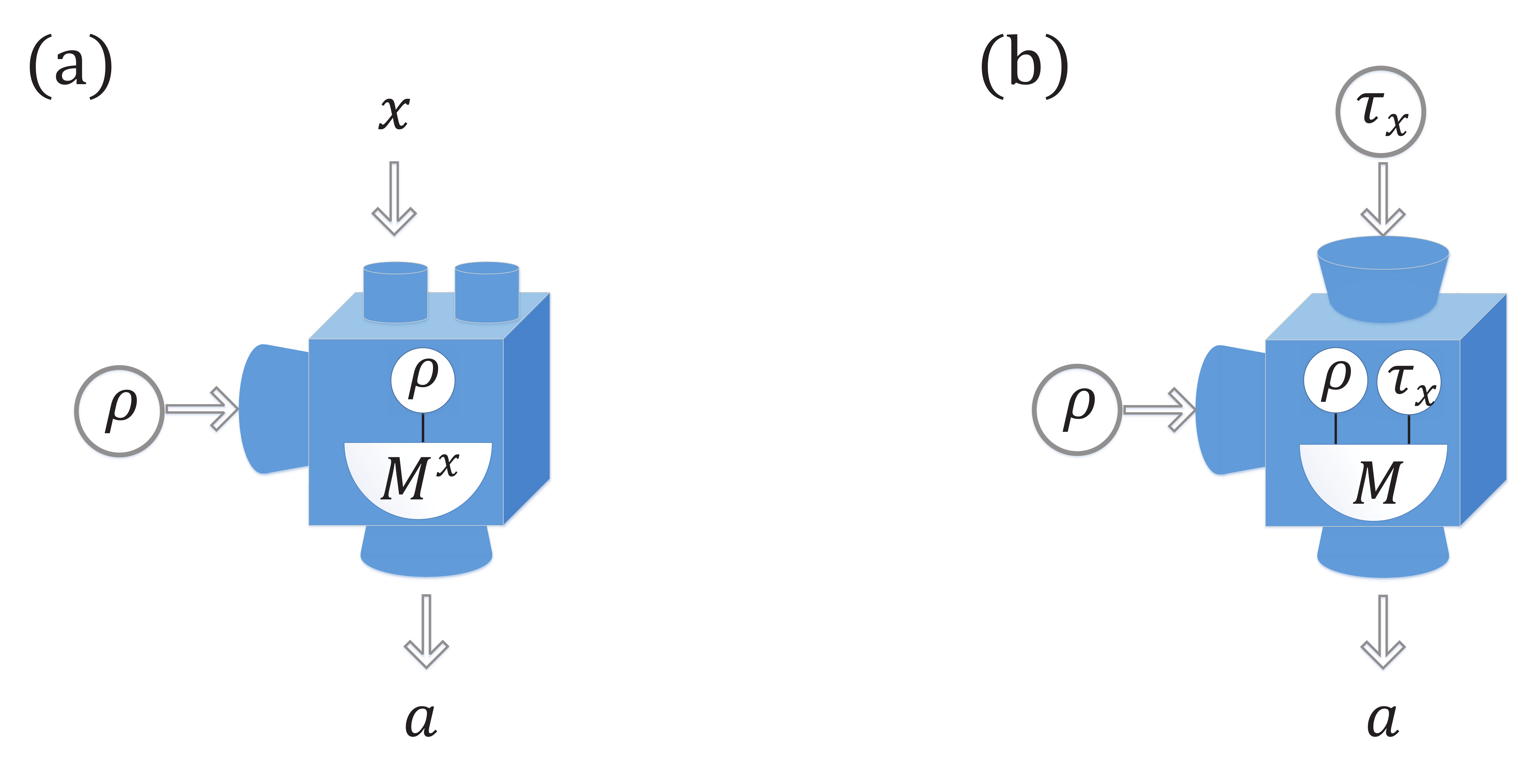}}
\caption{(Infeasible) DI coherence witness without and with ancillary states. (a) Fully DI scenario without ancillary states. The classical input $x$ determines a POVM $M^x$ (no characteristic except its label is known) for the measurement on the unknown quantum state $\rho$. The output $a$ and its probability distribution $p(a|x)$ are all that we are accessible to. (b) Semi DI scenario with ancillary states. In this scenario, apart from the unknown state $\rho$, some known (trusted) ancillary states $\tau_x$ are sent to the device as well. We describe the measurement process as a joint measurement on $\rho$ and $\tau_x$ (which is independent of label $x$), and the only information we have is the probability distribution $p(a|x)$, where $x$ denotes the label of the ancillary state $\tau_x$.} \label{Fig:NOGO}
\end{figure}

\subsection{Fully device-independent test}
First, we consider a fully DI scenario as shown in Fig.~\ref{Fig:NOGO}(a). Here, the untrusted device has an input $x$ and an output $a$, which gives a probability distribution $p(a|x)$. More than one POVM is possible, and the input $x$ determines by which POVM, $M^x$, the unknown quantum state is measured. We prove that this device cannot be used to witness coherence solely based on the probability distribution $p(a|x)$. For an unknown state $\rho$, the probability is given by
\begin{equation}\label{}
  p(a|x) = \Tr[\rho M^{x}_a],
\end{equation}
where $M^{x}_a$ is the element of POVM $M^x$ yielding the result $a$.

The probability distribution given by an incoherent state $\sigma = \sum_{i}p_i\ket{i}\bra{i}$ measured by $N^{x}=\{N^{x}_a\}$ is
\begin{equation}\label{}
    p_{\sigma}(a|x) = \sum_{i}p_i\Tr[\ket{i}\bra{i} N^{x}_a].
\end{equation}
Then we want to prove that any probability distribution $p(a|x)$ can be recovered by measuring incoherent states. Suppose the incoherent state is $\sigma = \ket{0}\bra{0}$ and the measurement operator is $N^{x}_a = p(a|x)\mathbb{I}$. It is easy to verify that $\{N^{x}_a\}$ form a POVM $N^{x}$ and $p(a|x) = p_{\sigma}(a|x)$.

A more careful thought on the definition of coherence displays the infeasibility of this scenario as well. Different from the problems about entanglement, we always need to appoint a certain computational basis when referring to quantum coherence. Yet one major characteristic of a fully DI scenario is its lack of a reference. It is therefore quite problematic to witness coherence under a certain computational basis via a fully DI method.

\subsection{Semi device-independent test: a joint-measurement scenario}
Now, we consider the case where the classical inputs are replaced by quantum inputs as shown in Fig.~\ref{Fig:NOGO}(b). That is, instead of inputting $x$, we input a quantum state $\tau_x$. Then the probability distribution is given by
\begin{equation}\label{}
  p(a|x) = \Tr[(\rho\otimes\tau_x) M_a],
\end{equation}
where $M_a$ is a POVM element that acts on $\rho$ and $\tau_x$, yielding the result $a$. The fully DI scenario in  Fig.~\ref{Fig:NOGO}(a) is a special case of  the scenario in Fig.~\ref{Fig:NOGO}(b), since letting $\tau_x = \ket{x}\bra{x}$ we will have the fully DI case. The extra advantage with ancillary states is to exploit the feature of imperfect distinguishability of non-orthogonal states. However, we will prove that the semi DI scenario with ancillary states cannot witness coherence either.

The probability distribution given by an incoherent state $\sigma = \sum_{i}p_i\ket{i}\bra{i}$ measured by $N = \{N_a\}$ is
\begin{equation}\label{}
    p_{\sigma}(a|x) = \sum_{i}p_i\Tr[(\ket{i}\bra{i}\otimes\tau_x) N_a],
\end{equation}
where $N_a$ is a POVM element that acts on $\sigma$ and $\tau_x$. Then, we can also show that the probability distribution with the incoherent state can recover all probability distributions. Given the spectral decomposition $\rho = \sum_{i}\lambda_i\ket{\psi_i}\bra{\psi_i}$, we can find an incoherent state $\sigma = \sum_{i}\lambda_i\ket{i}\bra{i}$ and POVM measurement $N_a =\sum_{j} (\ket{j}\bra{\psi_j}\otimes \mathbb{I})M_a(\ket{\psi_j}\bra{j}\otimes \mathbb{I})$, such that $p(a|x) = p_{\sigma}(a|x)$.

\section{Semi device-independent scenario with a prepare-and-measure set-up}\label{newsemi}
\subsection{Set-up of the scenario}
In the existing DI scenarios analyzed in Section~\ref{nogo}, we fail to gain information about the untrusted device via a joint measurement on $\rho$ and $\tau_x$ whether the inputs $\tau_x$ are seen as classical or quantum. Intuitively, with some trustworthy ancillary quantum states available, we can use them to obtain information about the measurement device first. This idea has been used for measurement tomography and randomness generation~\cite{CaoMDIQRNG,Tavakoli2018selftesting,PhysRevA.94.060301}, while we take one step further to see if we can witness and quantify the unknown state's coherence using a similar approach.

As shown in Fig.~\ref{Fig:NewSEMIDI}, instead of jointly measuring $\rho$ and $\tau_x$ in Fig.~\ref{Fig:NOGO}(b), we randomly input $\rho$ or $\tau_x$. That is, suppose the input set is $S = \{\rho_x|\rho_0=\rho,\,\rho_i=\tau_i,\,i = 1, 2, ... \}$, we randomly input $\rho_x\in S$.
In the i.i.d. limit we can treat the measurement as a fixed one. In addition, we assume that $\rho$ lies in the subspace spanned by the ancillary states, which will be explained in detail in the following.

\begin{figure}[htp]
\centering \resizebox{8cm}{!}{\includegraphics{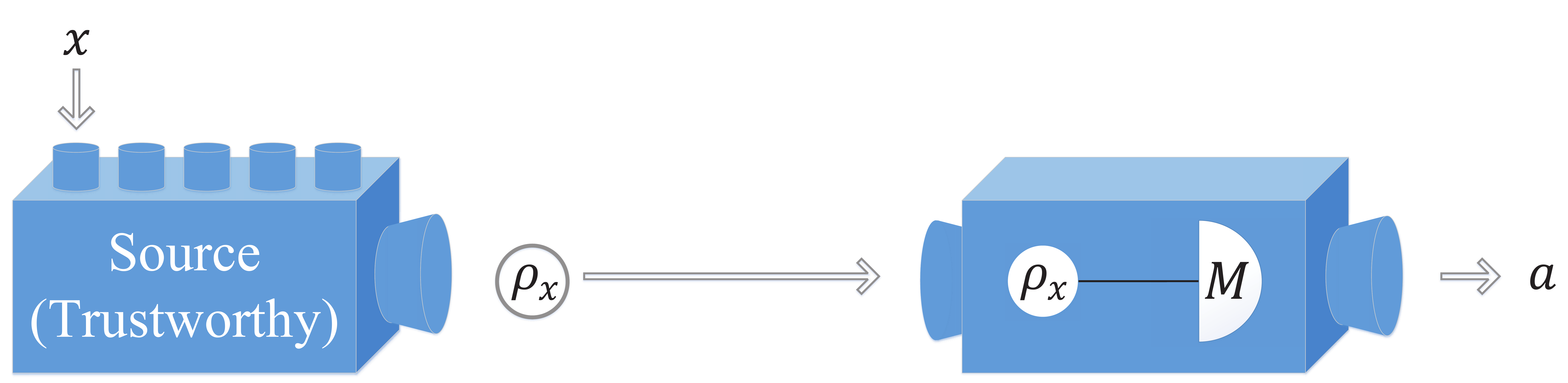}}
\caption{Coherence witness with ancillary states via a prepare-and-measure scenario. Based on a random classical input $x$, we send the unknown quantum state $\rho$ or one of the trusted ancillary states $\tau_x$ to the untrustworthy measurement device. Some (fixed) measurement on the quantum state is performed and a classical output $a$ is generated.} \label{Fig:NewSEMIDI}
\end{figure}

\subsection{Mathematical description of coherence quantification problem}
With our new semi DI scenario, we now ask the question of coherence quantification, which is to lower bound the coherence of an unknown state. This is a stronger problem than witnessing coherence. Here we consider the simplest scenario that there are only two outcomes, which are determined by a POVM that consists of two elements $M_1, M_2$. Since the two POVM elements should satisfy the completeness relation, we only need to take one element into consideration, say, $M_1$. In the following we denote $M_1$ as $M$ for convenience, and omit the subscript of its corresponding measurement result, unless specified otherwise. The coherence quantification problem is then stated as follows:

\emph{Problem1: In an appointed computational basis, given an unknown quantum state $\rho$ and an unknown POVM $\{M_,\,\mathbb{I}-M\}$, find}
\begin{align*}
&\min_{\rho}\quad C(\rho),\\
&\begin{array}{r@{\quad}r@{}l@{\quad}l}
\text{s.t. }&\Tr[\rho M]&=m,\\
            &\Tr[\tau_x M]&=n_x,\,x=1, 2, \ldots.\\
\end{array}
\end{align*}
$m,\,n_x$ are known statistics, $\tau_x$ are known, trusted ancillary states, and $\rho,\,M$ are the unknown quantum state and POVM element, respectively. $C$ is a certain functional which is a valid coherence measure.
In the following, we assume that coherence is defined in the computational basis $J=\{\ket{0},\ket{1},\ldots,\ket{d-1}\}$.

With ancillary states $\{\tau_x\}$ that form an informationally complete basis, we can carry out a full measurement tomography to determine all the POVM elements in the subspace spanned by these states. As $\rho$ also lies in this subspace, it is plausible to bound its coherence from the measurement statistics. If $\rho$ has components orthogonal to the subspace, as we know nothing about the POVM's behaviour on these components, we cannot decide whether $\rho$ is coherent or not from the statistics. This assumption is reasonable in many scenarios. For example, in a QKD experiment with a well-characterised photon source, a filter is used to project the photon into a qubit. We can reasonably assume that states after the filter lie in the same squashed space.

With known POVM elements, the original problem becomes an optimization problem with linear constraints:

\emph{Problem2: Given an unknown quantum state $\rho$ and a known POVM $\{M_,\,\mathbb{I}-M\}$, find}
\begin{align*}
&\min_{\rho}\quad C(\rho),\\
&\begin{array}{r@{\quad}r@{}l@{\quad}l}
\text{s.t. }&\Tr[\rho M]&=m.\\
\end{array}
\end{align*}
\emph{Problem2} is much easier than \emph{Problem1}, since only the quantum state $\rho$ is unknown, and the problem is a convex one. Yet some useful information still can be gained even if only a partial measurement tomography is made. In the following sections, we'll mainly focus on the case where a full measurement tomography is provided. Afterwards some discussions on \emph{Problem1} are made as well.

Before tackling the optimization problems we briefly show how a POVM tomography can be done. First we consider the simple qubit case. It is well known that a qubit can be expressed with Pauli matrices
\begin{equation}\label{qubit}
    \rho=\frac{\mathbb{I}+\vec{r}\cdot\vec{\sigma}}{2},
\end{equation}
where $\vec{r}$ is a three-dimensional real vector with its length $\|\vec{r}\|_2\in[0,1]$, $\mathbb{I}$ is the identity matrix, and $\vec{\sigma}=(\sigma_x,\sigma_y,\sigma_z)$ is the vector of Pauli matrices. Similarly, a two-dimensional POVM can be expressed in the form
\begin{equation}\label{POVM}
\begin{aligned}
& M_1=a(\mathbb{I}+\vec{\nu}\cdot\vec{\sigma}),\\
& M_2=\mathbb{I}-M_1,
\end{aligned}
\end{equation}
where $\vec\nu$ is a three-dimensional real vector~\cite{Kurotani07}. According to the definition of POVM, the parameters $a$ and $\vec\nu$ are such that $M_1,M_2 \succeq 0$, hence we have
\begin{equation}
\begin{aligned}
    0&<a<1,\\
    \|\vec\nu\|_2 &\leq \min \{1,\,\frac{1-a}{a}\}.
\end{aligned}
\end{equation}

If we input four ancillary states which form a complete basis, for instance, $\tau_x \in \{\ket{0},\ket{1},\ket{+}=1/\sqrt{2}(\ket{0}+\ket{1}),\ket{+i}=1/\sqrt{2}(\ket{0}+i\ket{1})\}$, we can do a full measurement tomography of the POVM. That is,
\begin{equation}\label{tomography}
\begin{aligned}
   &\left\{
        \begin{array}{lr}
    p(1|\ket{0})=a(1+\nu_z),\\
    p(1|\ket{1})=a(1-\nu_z),\\
    p(1|\ket{+})=a(1+\nu_x),\\
    p(1|\ket{+i})=a(1+\nu_y).
    \end{array}
          \right.
\end{aligned}
\end{equation}
From the mathematical perspective, it is straightforward to see that the POVM element $M_1$ can be exactly determined by solving the set of equations in Eq.~\eqref{tomography}, and $M_2$ can be determined from the completeness relation afterwards (in a real experiment, however, an employment of a maximum likelihood estimation method is preferred to directly solving Eq.~\eqref{tomography}, while this is not the main point of the coherence witness and quantification problem and is beyond the scope of this paper). It should be noted that the information we gained about the POVM is actually restricted to a subspace spanned by ancillary states. In this sense we will implicitly make statements like qudit condition and $d$-dimensional system, which actually refers to the dimension of the states and the subspace of the POVM investigated.

The discussion can be generalized to a $d$-dimensional system and more measurement outcomes. Notice the fact that any valid density matrices and POVM elements are Hermitian operators, and thus can be expressed as a linear combination of identity operator and the standard generators of SU(d) algebra~\cite{qudit-tomography}
\begin{equation}\label{qudit}
\rho=\frac{\mathbb{I}+\sum_{i=1}^{d^2-1} r_i \hat{\lambda}_i}{d},
\end{equation}
\begin{equation}\label{qudit POVM}
M_j= a_j\left(\mathbb{I}+\sum_{i=1}^{d^2-1} \nu_{j,i} \hat{\lambda}_i\right),
\end{equation}
where $d$ denotes the dimension of the Hermitian space and $\{\hat{\lambda}_i\}$ are the standard generators of SU(d) algebra. The construction of $\{\hat{\lambda}_i\}$ can be found in~\cite{qudit-tomography,qudit-Bloch-vectors}. We mainly use the notations in~\cite{qudit-tomography}, and we present a brief review on this in Appendix~\ref{generators}. In Eq.~\eqref{qudit}, the coefficients $r_i$ form a generalized Bloch vector $\vec{r}$ in $d$-dimensional Hilbert space with its length $\|\vec{r}\|_2\in[0,\sqrt{\frac{d(d-1)}{2}}]$. $\rho$ and $M_j$ are positive operators, and $\sum_j{M_j}=\mathbb{I}$. The measurement tomography can be carried out similarly to the qubit POVM, since we have
\begin{equation}\label{PR}
\mathrm{Tr}[\rho M_j]=a_j(1+\frac{2}{d}\vec{r}\cdot\vec{\nu}_j).
\end{equation}
In the following sections, we mainly focus on the qubit case with binary outcomes, and some characteristics specific to high dimensional cases are discussed afterwards.

\section{Numerical approaches to a lower bound for $C_{RE}$}\label{num}
First we take the relative entropy of coherence as the coherence measure.
{We choose this measure for its extensive use in quantum information processing. }
\emph{Problem 2} is then as follows:

\emph{Problem2(a): Given an unknown quantum state $\rho$ and a known POVM $\{M_,\,\mathbb{I}-M\}$, find}
\begin{align*}
&\min_{\rho}\quad C_{RE}(\rho),\\
&\begin{array}{r@{\quad}r@{}l@{\quad}l}
\text{s.t. }&\Tr[\rho M]&=m.\\
\end{array}
\end{align*}
We give two numerical methods for this problem:

\textbf{Method 1}: Convex optimization with linear constraints

$C_{RE}(\rho)$ is a convex function with respect to $\rho$ due to the joint convexity of the relative entropy, and all quantum states satisfying the constraint form a convex set, making \emph{Problem 2(a)} a convex optimization problem. In addition, the constraint we have here is linear. We can express the density matrices using Bloch vectors, and derive another optimization problem in the real vector space. Remarkably, equivalence between the representations of density matrices and Bloch vectors holds only in qubit case. In higher dimensional cases, not all matrices in the form of Eq.~\eqref{qudit} are valid quantum states, which we'll discuss in Section~\ref{Cl1} in detail.

In qubit case, the equivalent optimization problem in the space of $\mathbb{R}^3$ is as follows

\begin{align*}
&\min_{\vec{r}}\quad S(\rho(\vec{r})||\Delta(\rho(\vec{r}))),\\
&\begin{array}{r@{\quad}r@{}l@{\quad}l}
\text{s.t. }&\Tr[M]+\sum_{i}r_i\Tr[\sigma_i M]=2m,\\
\end{array}
\end{align*}
where $S(\rho||\Delta(\rho))$ is the relative entropy of coherence of $\rho$, and $\rho$ is related with $\vec{r}$ through the expression Eq.~\eqref{qubit}.
We can use some off-the-shelf numerical packages to solve this optimization problem with accuracy and high efficiency.

The problem turns out to be much more difficult if only a partial measurement tomography can be made, due to the quadratic form of the constraint $\Tr[\rho M]=m$. While inspired by the representation using Bloch vectors, we can at least employ a brutal-force numerical method. Noticing that when we express  the POVM element $M$ in the way of Eq.~\eqref{POVM}, we require $\|\vec\nu\|_2\leq \min \{1,\,\frac{1-a}{a}\}$. We can go over the region determined by the set of all possible POVMs with some appropriate sampling, and for each sampled point, we solve an optimization problem in the form of \emph{Problem 2(a)}. We can come to a result by comparing the optimal values at each sampled point. Cumbersome as it is, this method can give us the ``best'' result in theory, since we make no approximation apart from sampling (some approximation might be made in the algorithm employed by the numerical package, though). We can regard the result given by this method as a ``standard'' one.

~\\

\textbf{Method 2}: Optimization with Lagrange duality

Apart from the directly-solving method, we introduce an optimization method in~\cite{COQKD} based on Lagrange duality. The optimization satisfies the strong duality criterion and therefore we can consider its dual problem
\begin{equation}\label{dual}
\beta:=\max_{\lambda}\min_{\rho}\frac{1}{\ln 2} \mathcal{L} {(\rho, \lambda)},
\end{equation}
where the Lagrangian $\mathcal{L}$ is
\begin{equation}\label{Lagrangian}
    \mathcal{L} {(\rho, \lambda)}=S\left(\rho||\Delta(\rho)\right)+\lambda(\Tr[\rho M]-m),
\end{equation}
$\lambda$ is the introduced Lagrangian multiplier. $\beta$ is the optimal value of the dual problem, and strong duality implies that it is also the optimal value of the primal problem.

Using a property of $C_{RE}$
\begin{equation}
C_{RE}(\rho)=\min_{\delta\in\mathcal{I}}S(\rho||\delta)=S(\rho||\Delta(\rho)), \end{equation}
where $\mathcal{I}$ denotes the set of all incoherent states under the appointed computational basis, we can construct another function by introducing a new variant density matrix $\sigma$
\begin{equation}\label{}
    f(\rho,\sigma,\lambda):=S(\rho||\Delta(\sigma))+\lambda (\Tr[\rho M]-m),
\end{equation}
and re-express the dual problem in the form of a three-level optimization problem
\begin{equation}\label{eq:hatbeta}
    \beta=\frac{1}{\ln2}\max_\lambda\min_{\rho}\min_{\sigma}f(\rho,\sigma,\lambda).
\end{equation}
The two minimizations in Eq.~\eqref{eq:hatbeta} can be interchanged. We first solve $\min_{\rho}f(\rho,\sigma,\lambda)$, acquiring the unique solution and the optimal value
\begin{equation}\label{}
    \rho^*=\mathrm{exp}[-\mathbb{I}-\lambda M+\mathrm{ln}(\Delta(\sigma))],
\end{equation}
\begin{equation}\label{}
    f(\rho^*,\sigma,\lambda)=-\Tr[\rho^*]-\lambda m.
\end{equation}
We then employ Golden-Thompson inequality,
and obtain a lower bound on the optimal value
\begin{equation}\label{lowerbound}
    \beta\geq\frac{1}{\ln2}\max_{\lambda}\left(-\|\Delta(\exp(-\mathbb{I}-\lambda M))\|_{\infty}-\lambda m\right),
\end{equation}
where $\|\cdot\|_{\infty}$ denotes the maximum eigenvalue. In conclusion, we use Eq.~\eqref{lowerbound} as a lower bound for the coherence of the unknown state.

The advantage of this method is that the number of free parameters in the optimization is equal to the number of constraints and hence independent of the system's dimension. When the POVM contains only two elements as in our case, we only have one linear constraint if a full measurement tomography is made. Thus, this method will become more efficient when the dimension goes very large. On the other hand, however, the estimation result is not tight due to the use of Golden-Thompson inequality. Besides, as a duality approach is used, this method fails in the circumstance where only a partial measurement tomography is made.

Here we give some specific examples in the qubit case using our numerical methods in order to demonstrate some characteristics of the coherence witness and quantification problem.
Suppose after a full measurement tomography, we learn that for the POVM element $M=a(\mathbb{I}+\vec{\nu}\cdot\vec{\sigma})$, $a=0.6,\,\vec{\nu}=(1/2,\,1/4,\,1/4)$. Consider an unknown quantum state in the form $\rho=\frac{\mathbb{I}+q\vec{r}\cdot\vec{\sigma}}{2}$, where $q$ is an unknown parameter ranging in $[0,1]$, and $\vec{r}$ is some three-dimensional vector of which the length is $1$. We can interpret $q$ as a parameter denoting the state's purity. For $\vec{r_1}=(1,\,0,\,0)$ and $\vec{r_2}=(0,\,1,\,0)$, we compute the lower bound of the state's relative entropy of coherence when the value of $q$ changes from $0$ to $1$, as shown by Fig.~\ref{Fig:re1} and Fig.~\ref{Fig:re2}, respectively. Besides the results derived from two numerical methods, we also give the actual relative entropy of coherence for comparison, as shown by the blue solid lines.

\begin{figure}[htp]
\centering \resizebox{8cm}{!}{\includegraphics{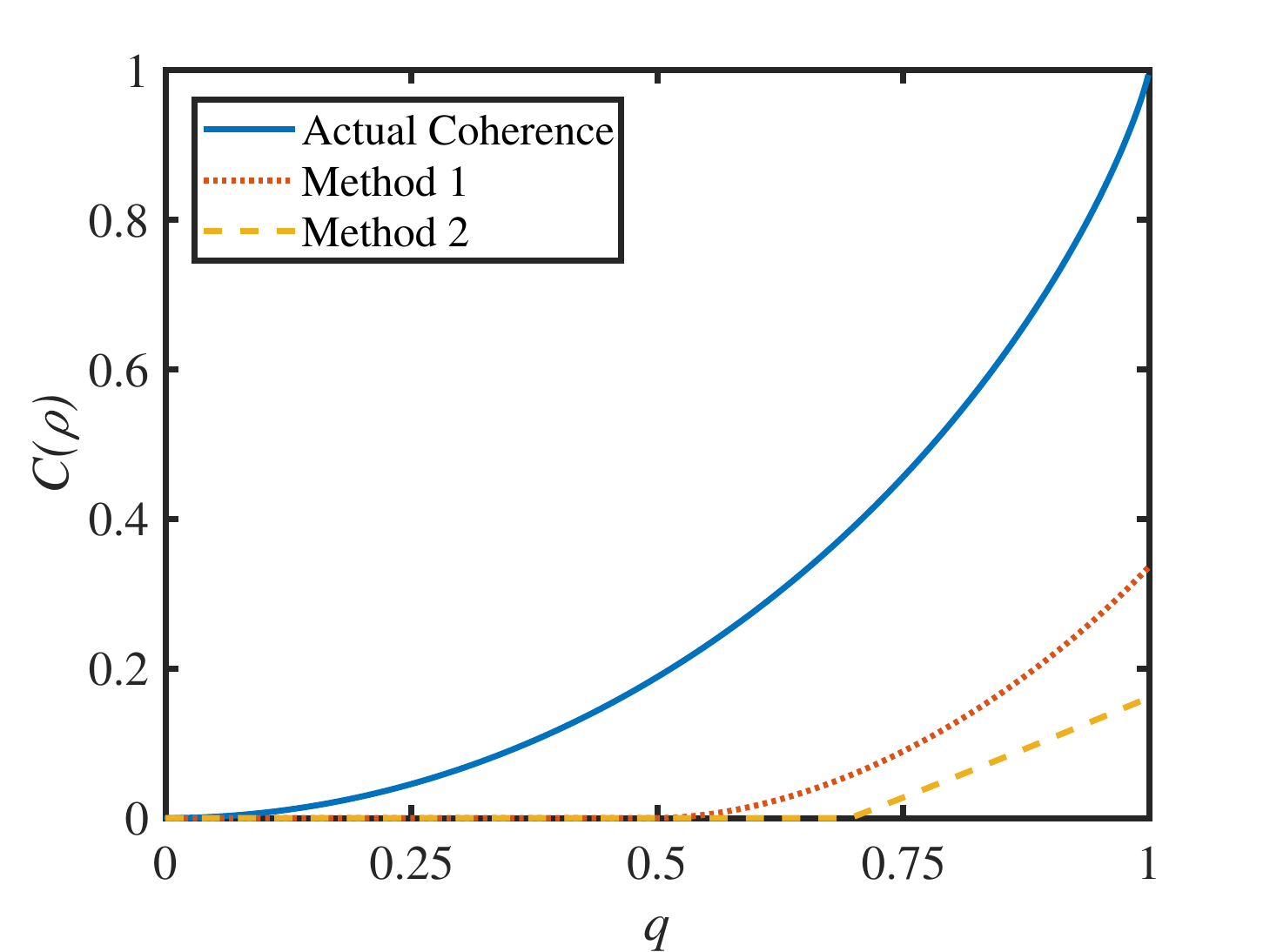}}
\caption{The estimation results via two numerical approaches and the actual value of $C_{RE}$ of the state $\rho=\frac{\mathbb{I}+q\vec{r_1}\cdot\vec{\sigma}}{2}$, with $q\in[0,1]$ and $\vec{r_1}=(1,\,0,\,0)$. The state is coherent under $z$-basis when $q>0$, yet the numerical estimation methods are able to give a non-zero lower bound on coherence only when $q$ is above some certain threshold: for Method 1 the threshold is 0.5, and for Method 2 the threshold is about 0.7. Generally, Method 1 yields a better result than Method 2.} \label{Fig:re1}
\end{figure}

\begin{figure}[htp]
\centering \resizebox{8cm}{!}{\includegraphics{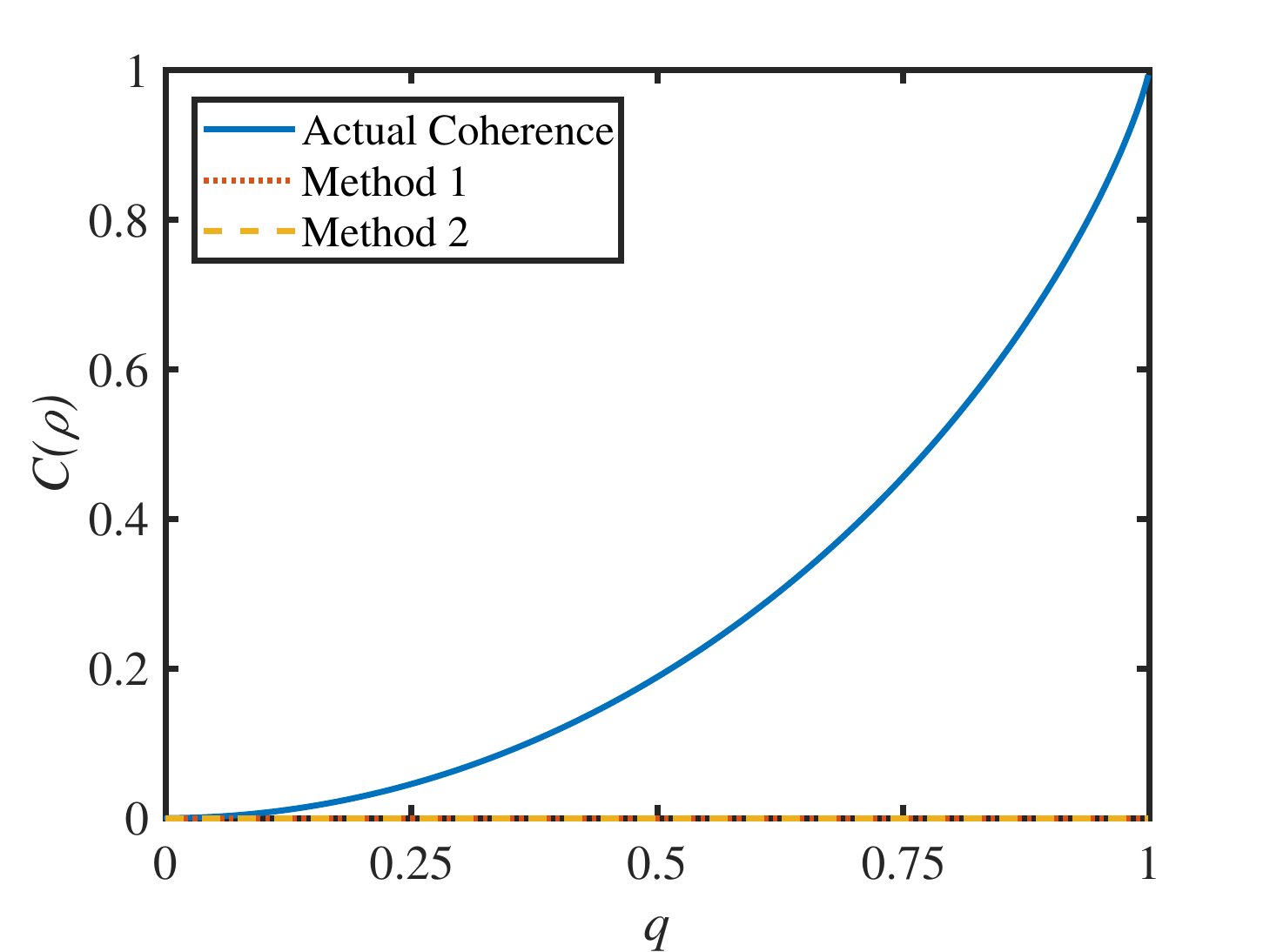}}
\caption{The estimation results via two numerical approaches and the actual value of $C_{RE}$ of the state $\rho=\frac{\mathbb{I}+q\vec{r_2}\cdot\vec{\sigma}}{2}$, with $q\in[0,1]$ and $\vec{r_2}=(0,\,1,\,0)$. The state is coherent under $z$-basis when $q>0$, yet neither of the two numerical estimation methods are able to give a non-zero lower bound on coherence, whatever the value $q$ takes.} \label{Fig:re2}
\end{figure}

From Fig.~\ref{Fig:re1}, we see that the two methods generally yield a valid (above zero) coherence quantification result, and the curve representing the first method is above the curve representing the second one, which accords with our analysis. Yet for the state $\rho=\frac{\mathbb{I}+q\vec{r_1}\cdot\vec{\sigma}}{2}$, we find that only when $q>0.5$ are we able to give a non-zero lower bound to its relative entropy of coherence, even with the first method. Method 2 requires an even larger threshold, where it gives a non-zero lower bound when $q$ is no less than about 0.7. Besides, under the computational basis $J=\{\ket{0},\,\ket{1}\}$, the states $\rho=\frac{\mathbb{I}+q\vec{r_1}\cdot\vec{\sigma}}{2}$ and $\rho=\frac{\mathbb{I}+q\vec{r_2}\cdot\vec{\sigma}}{2}$ hold the same non-zero $C_{RE}$ for identical $q$. Yet we see that whatever the value $q$ takes, we cannot validly bound the relative entropy of coherence of $\rho=\frac{\mathbb{I}+q\vec{r_2}\cdot\vec{\sigma}}{2}$.
In the next section where $l_1$ norm of coherence is used as the coherence measure, similar ``failures'' exist as well, while there we will analytically show that they are due to a special combination of the state and measurement.
%

\section{A tight analytical lower bound for $C_{l_1}$}\label{Cl1}
Although the numerical methods are easy to be implemented, they do not give us a clear picture with an intuitive physical interpretation. For instance, we cannot tell when we can achieve a result equal to the actual coherence a quantum state holds. In addition, it is not clear why we cannot estimate a state's coherence sometimes. For this reason, we hope to derive an analytical method for this problem.
As it is hard to derive an analytical lower bound for the relative entropy of coherence, we consider the $l_1$ norm for coherence for this task, which has a quite simple mathematical form.
The problem is as follows when we apply the $l_1$ norm of coherence as the coherence measure:

\emph{Problem2(b): Given an unknown quantum state $\rho$ and a known POVM $\{M_,\,\mathbb{I}-M\}$, find}
\begin{align*}
&\min_{\rho}\quad C_{l_1}(\rho),\\
&\begin{array}{r@{\quad}r@{}l@{\quad}l}
\text{s.t. }&\Tr[\rho M]&=m.\\
\end{array}
\end{align*}

As in the scenario when applying $C_{RE}$ as the coherence measure, we have a convex optimization problem which can be accurately and efficiently solved numerically using some off-the-shelf softwares. We take the result derived in this way as a ``standard'' result. Now we show how a tight analytical bound can be achieved. Still, we first consider the qubit case. Under $Z$-basis, the $l_1$ norm of coherence for a state $\rho=\frac{\mathbb{I}+\vec{r}\cdot\vec\sigma}{2}$ is given by~\cite{Xiao}
\begin{equation}\label{}
    C_{l_1}(\rho) = \sqrt{r_x^2+r_y^2}.
\end{equation}

The optimization problem with the qubit $\rho$ can be transferred into an equivalent problem on its corresponding Bloch vector $\vec{r}$. In qubit case, the following statements are equivalent:

\begin{align*}
&(\mathrm{S}1)\,\rho\succeq0,\,\mathrm{Tr}[\rho]=1,\\
&(\mathrm{S}2)\,\rho=\frac{\mathbb{I}+\vec{r}\cdot\vec{\sigma}}{2},\,\|\vec{r}\|_2\leq 1.
\end{align*}
For higher dimensions ($d \geq 3$), generalized (S2):\,$\rho=\frac{\mathbb{I}+\vec{r}\cdot\vec{\hat{\lambda}}}{d},\,\|\vec{r}\|_2\leq \sqrt{\frac{d(d-1)}{2}}$, which is derived from the requirement $\mathrm{Tr}[\rho^2]\leq 1$, is only a necessary condition for a valid quantum state. Here $\{\hat{\lambda}_i\}$ are standard generators of SU(d) algebra.
We can easily find counterexamples, e.g.
\begin{equation}
\rho=\left[
\begin{matrix}
\frac{1+\sqrt{3}}{3}&0&0&\\
0&\frac{1-\sqrt{3}}{3}&0&\\
0&0&\frac{1}{3}
\end{matrix}
\right].
\end{equation}
This matrix satisfies (S2), yet it is not positive, hence not a valid qutrit. This can cause some difficulties in the problem of coherence quantification in a general qudit case, as we will show later in this section.

Now let us return to the qubit case. 
The equivalent optimization problem is
\begin{align*}
&\min_{\vec{r}\in\mathbb{R}^3}\quad \sqrt{r_x^2+r_y^2},\\
&\begin{array}{r@{\quad}r@{}l@{\quad}l}
\text{s.t. }&a(1+\vec{\nu}\cdot\vec{r})&=m,\\
            &\|\vec{r}\|_2&\leq1.\\
\end{array}
\end{align*}
We now try to derive an analytical lower bound of coherence, beginning with the constraint given by the measurement result $\mathrm{Tr}[\rho M]=a(1+\vec{\nu}\cdot\vec{r})=m$.
For the target quantum state $\rho$, because $\|\vec{r}\|_2\leq1$, we have
\begin{equation}\label{scale1}
    \nu_z r_z\leq|\nu_z|\sqrt{1-(r_x^2+r_y^2)}.
\end{equation}
For the term $\nu_x r_x+\nu_y r_y$, we apply Cauchy-Schwarz inequality,
\begin{equation}\label{scale2}
\begin{aligned}
    \sqrt{(\nu_x^2+\nu_y^2)(r_x^2+r_y^2)}&\geq\sqrt{(|\nu_x r_x|+|\nu_y r_y|)^2}\\
    &=|\nu_x r_x|+|\nu_y r_y|\\
    &\geq \nu_x r_x+\nu_y r_y.
\end{aligned}
\end{equation}
Combining the inequalities Eq.~\eqref{scale1} and Eq.~\eqref{scale2} we come to the result
\begin{equation}\label{result1}
\begin{aligned}
&a\left(1+\sqrt{\nu_x^2+\nu_y^2}C_{l_1}(\rho)+|\nu_z|\sqrt{1-(C_{l_1}(\rho))^2}\right)\\
\geq&\,a(1+\vec\nu\cdot\vec{r})\\
=&\,m,
\end{aligned}
\end{equation}
where $C_{l_1}(\rho)=\sqrt{r_x^2+r_y^2}\in[0,1]$.
In the following, we assume that $\frac{m}{a}-1\geq0$. This assumption is reasonable, since there are two POVM elements, and according to the completeness relation we can always find one element that satisfies our assumption.

We take the smallest value satisfying this inequality, $C_{l_1}^{*}$, as the lower bound for coherence.
It is not obvious that we derive a lower bound from Eq.~\eqref{result1} for sure, as the inequality is quadratic essentially. Besides, from Eq.~\eqref{result1} it is not clear whether a non-zero bound for coherence can always be achieved as long as the POVM is a ``good'' one for coherence witness, that is, we cannot find an incoherent state to reconstruct the probability distribution. In addition, if a valid coherence bound can be achieved via our analytical approach, we naturally want to know whether it is tight. In other words, for any POVM with which we are able to witness coherence, can we always find a specific quantum state so that the lower bound of coherence equals to the actual coherence?

We prove that a lower bound can indeed be achieved in all circumstances with a ``good'' POVM from our analytical result, and our analytical approach is tight. We also discuss the physical meaning of the cases in which equality in Eq.~\eqref{result1} is achieved. Mathematically, we have the following conclusions:

\begin{theorem}\label{nonzerocondition}
  We cannot find an incoherent state to reconstruct the probability distribution derived from the measurement on a coherent state, if and only if $\nu_z^2< \left(\frac{m}{a}-1\right)^2$, and in this case we can always get a non-zero result by Eq.~\eqref{result1}. In other words, $C_{l_1}^{*}>0\Leftrightarrow\nu_z^2< \left(\frac{m}{a}-1\right)^2.$
\end{theorem}


\begin{proof}
In the computation basis $J$, a coherent state $\rho=\frac{\mathbb{I}+\vec{r}\cdot\vec{\sigma}}{2}$ has the property that $r_x^2+r_y^2\neq0$, while an incoherent state $\delta=\frac{\mathbb{I}+\vec{s}\cdot\vec{\sigma}}{2}=\frac{\mathbb{I}+s_z \sigma_z}{2}$. If we can reconstruct the probability by $\delta$, we have
\begin{equation}\label{fail}
a(1+\vec{\nu}\cdot\vec{r})=a(1+\nu_z s_z)=m.
\end{equation}
When $m=a$ and $\nu_z=0$, Eq.~\eqref{fail} is satisfied and we cannot witness coherence. Apart from this special condition, Eq.~\eqref{fail} is satisfied when $s_z=\frac{m-a}{a \nu_z}$. Since $|s_z|\leq1$, we can find an incoherent state $\delta$ to reconstruct the probability when $\nu_z^2\geq\left(\frac{m}{a}-1\right)^2$. So we cannot find an incoherent state to reconstruct the probability distribution if and only if $\nu_z^2< \left(\frac{m}{a}-1\right)^2$.

The second part of the proof is a mathematical deduction to simplify Eq.~\eqref{result1}, which we leave in Appendix~\ref{proof1}. Here we just present the simplified result. As long as $\nu_z^2< \left(\frac{m}{a}-1\right)^2$, we can derive a lower bound for $l_1$ norm of coherence of the unknown state, which is
\begin{equation}
    C_{l_1}^{*} = \frac{\left(\frac{m}{a}-1\right)\sqrt{\nu_x^2+\nu_y^2}-|\nu_z|\sqrt{\|\vec\nu\|_2^2-\left(\frac{m}{a}-1\right)^2}}{\|\vec\nu\|_2^2}.
\end{equation}
\end{proof}


\begin{theorem}\label{tight}
  When $\nu_z^2 < \left(\frac{m}{a}-1\right)^2$, we can always find a coherent state $\rho$ such that $C_{l_1}(\rho)=C_{l_1}^{*}$. In other words, our analytical approach is tight.
\end{theorem}


\begin{proof}

To prove this, we go back to the inequalities used in our approach, Eq.~\eqref{scale1}~\eqref{scale2}:

1. The condition required by a valid qubit density matrix:
\begin{equation*}\label{scale1proof}
    \nu_z r_z\leq|\nu_z||r_z|\leq|\nu_z|\sqrt{1-(r_x^2+r_y^2)},
\end{equation*}

2. Cauchy-Schwarz inequality:
\begin{equation*}\label{scale2proof}
\begin{aligned}
    \sqrt{(\nu_x^2+\nu_y^2)(\nu_x^2+\nu_y^2)}&\geq\sqrt{(|\nu_x r_x|+|\nu_y r_y|)^2}\\
    &=|\nu_x r_x|+|\nu_y r_y|\\
    &\geq \nu_x r_x+\nu_y r_y.
\end{aligned}
\end{equation*}

The conditions for equality are rather intuitive from the physical perspective. To achieve equality in Eq.~\eqref{scale1}, we need:

\emph{(a) $\rho$ is a pure state, thus $\|\vec{r}\|_2=1$. In the Bloch sphere representation, the Bloch vector $\vec{r}$ reaches the surface of the Bloch Sphere.}

\emph{(b) $\nu_z r_z\geq0$, i.e. in the Bloch Sphere representation, $\vec{r}$ and $\vec{\nu}$ are in the same semi-sphere (they are both in the north or the south).}

To achieve equality in Eq.~\eqref{scale2}, the projections of $\vec{\nu}$ and $\vec{r}$ on the $XY$-section of the Bloch sphere point to the same direction, that is:

\emph{(a) $\nu_x r_x\geq0,\,\nu_y r_y\geq0$.}

\emph{(b) $(\nu_x,\,\nu_y)$ and $(r_x,\,r_y)$ are in the same or opposite direction.}

We find that as long as a valid estimation is feasible, there always exists a quantum state which satisfies these conditions and generates the required probability distribution.
The detailed proof of this theorem is in Appendix~\ref{proof2}.
\end{proof}

To demonstrate our analytical method and compare its result with the actual coherence and the result derived by the numerical method, we use the same examples as in Sec.~\ref{num}, with results shown in Fig.~\ref{Fig:l1a}, \ref{Fig:l1b}.
In the case shown by Fig.~\ref{Fig:l1b}, and when $q\leq 0.5$ in the case shown by Fig.~\ref{Fig:l1a}, neither of the methods can give a non-zero lower bound for the state's coherence, which is in accordance with Theorem~\ref{nonzerocondition}. When $q>0.5$ in Fig.~\ref{Fig:l1a} where we can derive a valid lower bound, we see the estimation result derived from our analytical method coincides with the one obtained from the numerical method, showing tightness in our analytical result.


\begin{figure}[htp]
\centering \resizebox{8cm}{!}{\includegraphics{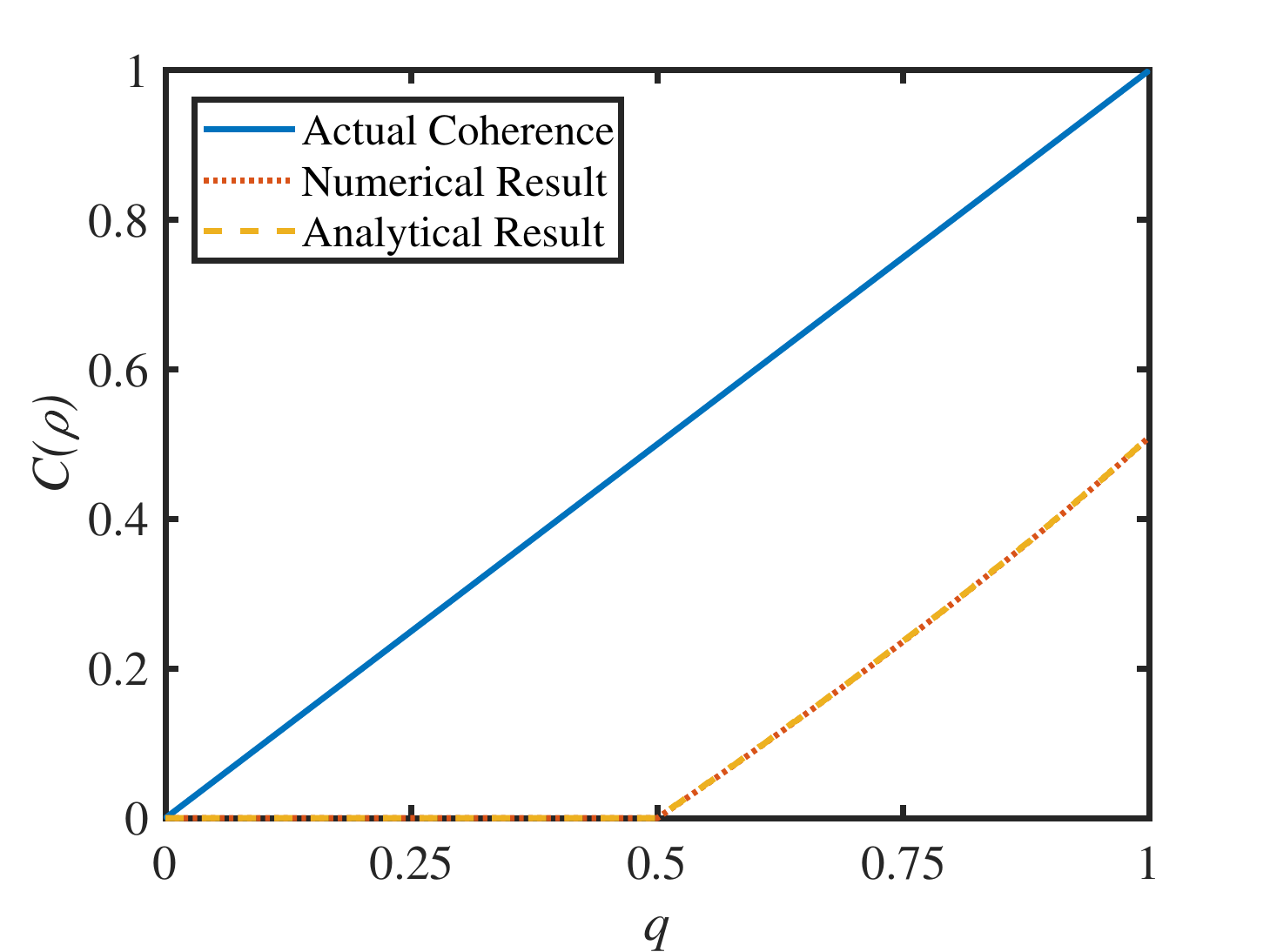}}
\caption{The estimation results via two numerical approaches and the actual value of $C_{l_1}$ of the state $\rho=\frac{\mathbb{I}+q\vec{r_1}\cdot\vec{\sigma}}{2}$, with $q\in[0,1]$ and $\vec{r_1}=(1,\,0,\,0)$. The analytical result coincides with the ``standard'' numerical result, giving a non-zero result when $q>0.5$.} \label{Fig:l1a}
\end{figure}

\begin{figure}[htp]
\centering \resizebox{8cm}{!}{\includegraphics{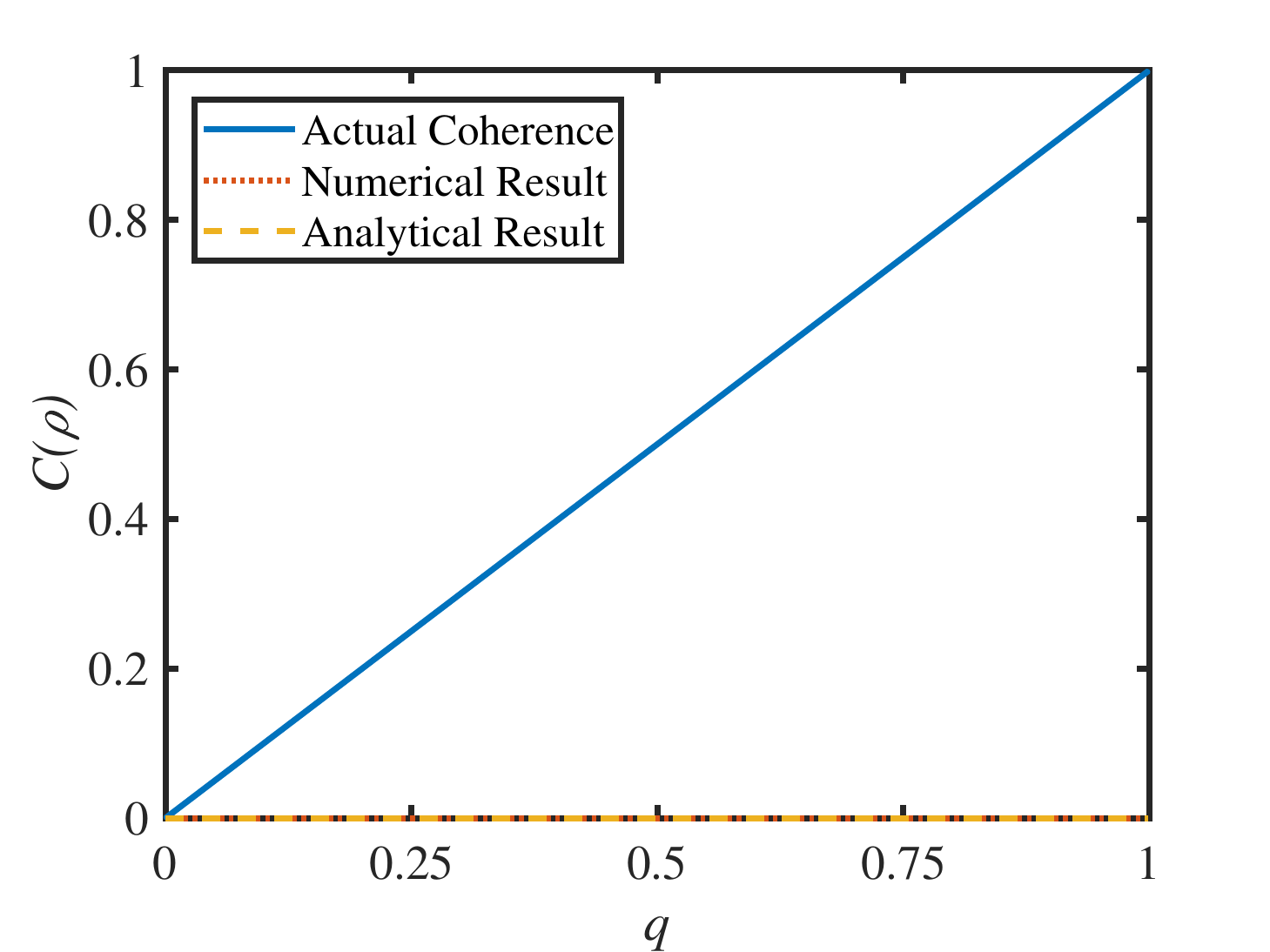}}
\caption{The estimation results via two numerical approaches and the actual value of $C_{l_1}$ of the state $\rho=\frac{\mathbb{I}+q\vec{r_2}\cdot\vec{\sigma}}{2}$, with $q\in[0,1]$ and $\vec{r_2}=(0,\,1,\,0)$. Neither of the methods can give a valid lower bound in this case.} \label{Fig:l1b}
\end{figure}

The analytical method can also be generalized to the high-dimensional case, where problems in higher dimensions follow a similar route.
Here we present a general result similar to Eq.~\eqref{result1}
\begin{widetext}
\begin{equation}\label{result3}
a\left(1+\mu\,C_{l_1}(\rho)+\frac{2}{d}\sqrt{\sum_{i=2}^{d}{\nu_{i^2-1}}}\sqrt{\frac{d(d-1)}{2}-\frac{d(C_{l_1}(\rho))^2}{2(d-1)}}\right)\geq a(1+\frac{2}{d}\vec\nu\cdot\vec{r})=m,
\end{equation}
\end{widetext}
where $\mu = \max_{1\leq j<i<d}\left\{\sqrt{\nu_{(i-1)^2+2(j-1)}^2+\nu_{(i-1)^2+2j-1}^2}\right\}$. In Appendix~\ref{qutrit} we show how to lower bound the $l_1$ coherence with a specific example of a qutrit. 
It's easy to prove that we indeed give a lower bound in Eq.~\eqref{result3} on a qudit's $l_1$ norm of coherence. This result, however, may suffer from the problem that the ``state" which takes the mark of equality is not a valid quantum state. This is due to the inequality scaling $\|\vec{r}\|_2\leq\sqrt{\frac{d(d-1)}{2}}$, where $\vec{r}$ is the generalized Bloch vector. Yet in high dimensional cases, as we've mentioned in this section previously, this is only a necessary condition for a valid quantum state. Thus we are not sure whether this approach is tight under high dimensional conditions. More knowledge about the algebra construction of a high dimensional quantum state is required.

In the condition of a partial measurement tomography, it's still possible to bound the $l_1$ norm of coherence of a quantum state in an inequality similar to Eq.~\eqref{result1} in form. Consider our qubit example, while this time only two ancillary states $\tau_1=\ket{0}\bra{0},\tau_2=\ket{1}\bra{1}$ corresponding to our appointed computational basis are available. $a,\,\nu_z$ in Eq.~\eqref{POVM} can be solved with these two ancillary states.
The positivity of POVM elements require that $\{M_1,M_2\}$ is a POVM, $M_1, M_2$ are both positive operators. From this constraint we have
\begin{equation}
\begin{aligned}
  &\left\{
        \begin{array}{lr}
    a(1-\|\vec{\nu}\|_2)\geq0,\\
    1-a(1+\|\vec{\nu}\|_2)\geq0.
    \end{array}
          \right.
\end{aligned}
\end{equation}
Solving the set of inequalities we have
\begin{equation}
    \nu_x^2+\nu_y^2\leq g(a)-\nu_z^2,
\end{equation}
where $g(a)=\min\{1,\frac{(1-a)^2}{a^2}\}$. Insert this inequality into Eq.~\eqref{result1}, $C_{l_1}(\rho)$ needs to satisfy
\begin{equation}
    a\left(1+\sqrt{g(a)-\nu_z^2}C_{l_1}(\rho)+|\nu_z|\sqrt{1-(C_{l_1}(\rho))^2}\right)\geq m.
\end{equation}

\section{Conclusion}
In this work, we introduce several scenarios of semi DI coherence witness and quantification with untrusted devices. We see that Bell tests or semi-quantum games cannot be applied to coherence witness in a straight-forward way, as Bell non-locality requires a stronger resource than single party quantum coherence.  In particular, contrary to the reference-independent feature of fully device-independent tests, coherence relies on the selection of the computational basis. It is thus impossible to witness and quantify single party coherence via a fully DI test. However, if we modify the way of using ancillary states in a semi-quantum game, that is, we first perform a measurement tomography via these trusted states, generally we can quantify an unknown state's coherence with untrusted measurement devices in a prepare-and-measure set-up. In this way we generalize the concept of (semi) DI tests to single party systems. In analyzing the feasibility of our new scenario, we consider the relative entropy of coherence and the $l_1$ norm of coherence as coherence measures. Thanks to the simple mathematical form of the $l_1$ norm of coherence, we find that we can quantify an unknown state's coherence with our new scenario, and we give a valid tight analytical method for estimation in the feasible range. This analytical approach can be naturally generalized to the high-dimensional case. As for the relative entropy of coherence, apart from a non-approximation numerical method, we borrow a Lagrangian-based duality numerical method, which was originally used in analyzing the key rate of quantum key distribution protocols. This method fails in the situation where only a partial measurement tomography is performed. While in the high-dimensional case with a full measurement tomography, this numerical approach becomes very efficient due to the use of Lagrangian duality. In all, our results can be seen as a further step on the discussion of prepare-and-measure approaches. Compared to previous works where the emphasis is on characterising unknown measurements~\cite{CaoMDIQRNG,Tavakoli2018selftesting,PhysRevA.94.060301}, we show the possibility on its application in witnessing quantum states.

Apart from the two distance-based coherence measures in this paper, there are many other coherence measures which may have clear operational interpretations in different situations, like the robustness of coherence \cite{ROC}, coherence of formation \cite{Xiao}, and the coherence measures related to one-shot dilution and distillation tasks \cite{zhao2018oneshot,Regula2018oneshot,zhao2018oneshotdistill}. We hope some explicit analytical results can be achieved for these measures in our new scenario. When estimating the relative entropy of coherence, we use the Golden-Thompson inequality. We notice that some improvements have been made on the original numerical estimation method in analyzing key rates of quantum key distribution protocols~\cite{Winick2018reliablenumerical}. We believe the improved method can be applied to the problem of coherence witness and quantification as well.
We hope our result may shed light on the close relation between coherence and randomness generation, and bring more possibilities in quantum cryptography.


\begin{acknowledgments}
We acknowledge P. Zeng, H. Zhou, Y. Cai and X. Ma for the insightful discussions. This work was supported by the National Natural Science Foundation of China Grants No. 11674193. XY is  supported by BP plc and by the EPSRC National Quantum Technology Hub in Networked Quantum Information Technology (EP/M013243/1).
\end{acknowledgments}

\appendix

\section{Construction of the Standard SU(d) Generators}\label{generators}
Here we briefly review the construction of the standard SU(d) generators $\hat{\lambda}_i$. First we introduce the elementary matrices of $d$ dimension, $\{e_i^j|i,j=1,\ldots,d\}$. $e_i^j$ denotes a matrix with its entry on the $j$th row, $i$th column equal to unity and all others equal to zero. With the help of these elementary matrices, we construct three types of traceless matrices:
\begin{equation}\label{realoff}
    \Theta_i^j = e_i^j + e_j^i,
\end{equation}
\begin{equation}\label{imagineoff}
    \beta_i^j = -i(e_i^j - e_j^i),
\end{equation}
\begin{equation}\label{diag}
    \eta_k^k = \sqrt{\frac{2}{k(k+1)}}\left[\sum_{i=1}^k e_i^i-ke_{k+1}^{k+1}\right],
\end{equation}
where in Eq.~\eqref{realoff}~\eqref{imagineoff} $1\leq j < i \leq d$, and in Eq.~\eqref{diag} $1 \leq k \leq d-1$. We can see that matrices $\Theta_i^j$ and $\beta_i^j$ are all Hermitian matrices with off-diagonal entries only, while matrices $\eta_k^k$ are diagonal real matrices. Then we define the SU(d) generators, or $\hat{\lambda}$ matrices,
\begin{equation}\label{type1}
    \hat{\lambda}_{(i-1)^2+2(j-1)} = \Theta_i^j,
\end{equation}
\begin{equation}\label{type2}
    \hat{\lambda}_{(i-1)^2+2j-1} = \beta_i^j,
\end{equation}
\begin{equation}\label{type3}
    \hat{\lambda}_{i^2-1} = \eta_{i-1}^{i-1}.
\end{equation}
We can see that Pauli matrices and Gell-Mann matrices are just the standard SU(d) generators in $2$- and $3$- dimension cases. $\hat{\lambda}_{(i-1)^2+2(j-1)}$ and $\hat{\lambda}_{(i-1)^2+2j-1}$ contribute to the off-diagonal terms, and $\hat{\lambda}_{i^2-1}$ contribute to the diagonal terms.

\section{Proofs of the Theorems}\label{proofs}
\subsection{Proof of Theorem 1}\label{proof1}
Now we modify Eq.~\eqref{result1} to this form
\begin{equation}\label{result1'}
    |\nu_z|\sqrt{1-(C_{l_1}(\rho))^2}\geq \left(\frac{m}{a}-1\right)-\sqrt{\nu_x^2+\nu_y^2}\,C_{l_1}(\rho).
\end{equation}

~\\

\emph{Case 1: $\frac{m}{a}-1\leq \sqrt{\nu_x^2+\nu_y^2}\,C_{l_1}(\rho)$}

In this case the inequality is established, and we derive a lower bound of coherence:
\begin{equation}\label{condition1}
    C_{l_1}^{*} = \frac{1}{\sqrt{\nu_x^2+\nu_y^2}}\left(\frac{m}{a}-1\right).
\end{equation}

~\\

\emph{Case 2: $\frac{m}{a}-1> \sqrt{\nu_x^2+\nu_y^2}\,C_{l_1}(\rho)$}

Squaring both sides of inequality Eq.~\eqref{result1'}, we get a quadratic inequality:
\begin{widetext}
\begin{equation}
    \|\vec{\nu}\|_2^2 C_{l_1}(\rho)^2-2\left(\frac{m}{a}-1\right)\sqrt{\nu_x^2+\nu_y^2}\,C_{l_1}(\rho)+\left[\left(\frac{m}{a}-1\right)^2-\nu_z^2\right]\leq 0.
\end{equation}
\end{widetext}
For this quadratic about $C_{l_1}(\rho)$, its discriminant is
\begin{equation}\label{discriminant1}
    \triangle=4\nu^2_z\left[\|\vec\nu\|_2^2-\left(\frac{m}{a}-1\right)^2\right]
\end{equation}
If $\triangle\geq0$, then
\begin{equation}\label{valid condition}
    \left(\frac{m}{a}-1\right)^2\leq \|\vec\nu\|_2^2,\,\text{or }\nu_z=0.
\end{equation}
From the measurement result, we see that $\vec{\nu}\cdot\vec{r}=\frac{m}{a}-1$, hence $(\vec{\nu}\cdot\vec{r})^2=\left(\frac{m}{a}-1\right)^2$. By applying Cauchy-Schwarz inequality, $(\vec{\nu}\cdot\vec{r})^2\leq \|\vec{\nu}\|_2^2 \|\vec{r}\|_2^2\leq \|\vec{\nu}\|_2^2$. Thus the case where $\triangle < 0$ is impossible.
We solve the quadratic inequality and derive the lower bound of coherence
\begin{equation}\label{condition2}
    C_{l_1}^{*} = \frac{\left(\frac{m}{a}-1\right)\sqrt{\nu_x^2+\nu_y^2}-|\nu_z|\sqrt{\|\vec\nu\|_2^2-\left(\frac{m}{a}-1\right)^2}}{\|\vec\nu\|_2^2}.
\end{equation}
The bound is valid when $\left(\frac{m}{a}-1\right)\sqrt{\nu_x^2+\nu_y^2} > |\nu_z|\sqrt{\|\vec\nu\|_2^2-\left(\frac{m}{a}-1\right)^2}$. If $\nu_z=0$, then we derive the same bound as given by Eq.~\eqref{condition1}. Except for this, simplifying this inequality, we have
\begin{equation}
    \left(\frac{m}{a}-1\right)^2 > \nu_z^2,
\end{equation}
which is the range within which we can possibly estimate a state's coherence, given by Theorem~\ref{nonzerocondition}.

By comparing the results in Case 1 and Case 2, we can see that the bound given by Eq.~\eqref{condition2} is no larger than Eq.~\eqref{condition1}. Therefore, the final estimation of coherence is Eq.~\eqref{condition2}.
And from the discussions above, we see that Eq.~\eqref{result1} gives us a lower bound of coherence indeed when a valid coherence bounding can be made. $\blacksquare$

~\\

\subsection{Proof of Theorem 2}\label{proof2}
Now we re-express the problem in the language of mathematics. We want to find a quantum state $\rho=\frac{\mathbb{I}+\vec{r}\cdot\vec\sigma}{2}$, subject to
\begin{equation}\label{tight problem}
\begin{aligned}
&\left\{
        \begin{array}{lr}
    a+a(\nu_x r_x+\nu_y r_y+\nu_z r_z)=m, \\
    \nu_x r_y=\nu_y r_x, \\
    r_x^2+r_y^2+r_z^2 = 1, \\
    \nu_x r_x,\,\nu_y r_y,\,\nu_z r_z \geq 0.
    \end{array}
          \right.
\end{aligned}
\end{equation}
$a,\,\vec{\nu}$ are already determined from previous measurement tomography. If such $\vec{r}$ can always be found when a valid coherence bound is made, we then prove our approach to be tight.

We notice the first two equations of Eq.~\eqref{tight problem} form a system of linear equations of $\vec{r}$
\begin{equation}\label{lin eq}
\begin{aligned}
&\left\{
        \begin{array}{lr}
\nu_x r_x+\nu_y r_y+\nu_z r_z =\frac{m}{a}-1, \\
\nu_y r_x-\nu_x r_y =0.
\end{array}
          \right.
\end{aligned}
\end{equation}
We can reasonably assume that $\nu_x^2+\nu_y^2 \neq 0$, since otherwise we can always find an incoherent state to recover measurement results, i.e. the POVM is a ``bad'' one.

~\\

\emph{Case 1: $\nu_z \neq 0$}

In this case, the general solution to Eq.~\eqref{lin eq} is
\begin{equation}
\begin{bmatrix}
  r_x \\
  r_y \\
  r_z \\
\end{bmatrix}
=
\lambda \begin{bmatrix}
          -\frac{\nu_x \nu_z}{\nu_x^2+\nu_y^2} \\
          -\frac{\nu_y \nu_z}{\nu_x^2+\nu_y^2} \\
          1 \\
        \end{bmatrix}
        +
        \begin{bmatrix}
          0 \\
          0 \\
          \frac{\frac{m}{a}-1}{\nu_z} \\
        \end{bmatrix}.
\end{equation}
$\lambda$ is an undetermined coefficient. In order that the last inequalities in Eq.~\eqref{tight problem} fit, we have
\begin{equation}\label{tight condition1}
\frac{m}{a}-1\geq -\lambda \nu_z \geq 0.
\end{equation}
and as we've already mentioned, $\frac{m}{a}-1 \geq 0$ always fits for one POVM element. To meet with the requirement of a pure state, we have
\begin{equation}
r_x^2+r_y^2+r_z^2=\lambda^2 \frac{\nu_z^2}{\nu_x^2+\nu_y^2}+\left(\lambda+\frac{\frac{m}{a}-1}{\nu_z}\right)^2=1.
\end{equation}
The discriminant of this quadratic about $\lambda$ is $\bigtriangleup=4[\frac{\|\vec{\nu}\|_2^2-\left(\frac{m}{a}-1\right)^2}{\nu_x^2+\nu_y^2}]$ ($\bigtriangleup$ here to be distinguished from that in Eq.~\eqref{discriminant1}), and by requiring $\bigtriangleup \geq 0$, we have
$\|\vec\nu\|_2\geq \frac{m}{a}-1$, which contains the range where we can validly estimate a state's coherence. Then we acquire the solution of $\lambda$
\begin{equation}
\lambda_{1,2}=\frac{-\frac{2(\frac{m}{a}-1)}{\nu_z}\pm\sqrt{\bigtriangleup}}{2\left(\frac{\|\vec{\nu}\|_2^2}{\nu_x^2+\nu_y^2}\right)}.
\end{equation}
It's easy to verify that in the range that a valid bound can be obtained, at least one of the solutions suffices Eq.~\eqref{tight condition1}.

~\\

\emph{Case 2: $\nu_z = 0$}

In this case, the solution to Eq.~\eqref{lin eq} is
\begin{equation}\label{tight condition2}
\begin{bmatrix}
  r_x \\
  r_y \\
\end{bmatrix}
=
\left(\frac{m}{a}-1\right)
\begin{bmatrix}
  \frac{\nu_x}{\nu_x^2+\nu_y^2} \\
  \frac{\nu_y}{\nu_x^2+\nu_y^2} \\
\end{bmatrix}.
\end{equation}
As for a valid quantum state, we require that $r_x^2+r_y^2 \leq 1$, hence for Eq.~\eqref{tight condition2}, we have $\left(\frac{m}{a}-1\right)^2 \leq \nu_x^2+\nu_y^2$, which is the same as Eq.~\eqref{valid condition}.
The solution suffices $\nu_x r_x,\,\nu_y r_y \geq 0$. We then acquire
\begin{equation}
r_z^2=\frac{\nu_x^2+\nu_y^2-\left(\frac{m}{a}-1\right)^2}{\nu_x^2+\nu_y^2}.
\end{equation}

In conclusion, we can always find a quantum state taking the mark of equality in Eq.~\eqref{result1}, as long as a valid coherence bounding can be made. Thus, we can say that our analytical approach is a tight one. $\blacksquare$

~\\

\section{An Analytical Result for a Qudit's $l_1$ Norm of Coherence (qutrit as an example)}\label{qutrit}
Using Gell-Mann matrices (the standard generators of SU(3) algebra), a qutrit can be represented as
\begin{equation}
\rho=\frac{\mathbb{I}+\sum_{i=1}^{8} r_i \hat{\lambda}_i}{3},\,\|\vec{r}\|_2\leq\sqrt{3}.
\end{equation}
The $\hat{\lambda}$ matrices are the same as in Appendix~\ref{generators}. To be more specific, the density matrix is
\begin{equation}
\begin{aligned}
 \rho = \frac{1}{3}
 \left[\begin{matrix}
    1+r_3+\frac{1}{\sqrt{3}}r_8 & r_1 - i r_2 & r_4 - i r_5 \\ r_1 + i r_2 & 1-r_3+\frac{1}{\sqrt{3}}r_8 & r_6 - i r_7 \\ r_4 + i r_5 & r_6 + i r_7 & 1 - \frac{2}{\sqrt{3}}r_8 \end{matrix}\right].
\end{aligned}
\end{equation}
By applying Eq.~\eqref{l_1 norm} we find the $l_1$ norm of coherence of a qutrit is
\begin{equation}
C_{l_1}(\rho)=\frac{2}{3}\left(\sqrt{r_1^2+r_2^2}+\sqrt{r_4^2+r_5^2}+\sqrt{r_6^2+r_7^2}\right).
\end{equation}
Suppose we have performed a full tomography on the qutrit POVM, that is, we have a full understanding about $M_j= a_j(\mathbb{I}+\sum_{i=1}^{8} \nu_{j,i} \hat{\lambda}_i)$. Still, we assume there are only two POVM elements, and we denote one element as $M$ and omit the subscripts of its coefficients. On the off-diagonal terms, we apply Cauchy-Schwarz inequality and obtain
\begin{equation}
\begin{aligned}
    &\nu_1 r_1+\nu_2 r_2+\nu_4 r_4+\nu_5 r_5+\nu_6 r_6+\nu_7 r_7\\
    \leq&\sqrt{\nu_1^2+\nu_2^2}\sqrt{r_1^2+r_2^2}+\sqrt{\nu_4^2+\nu_5^2}\sqrt{r_4^2+r_5^2}\\
    &+\sqrt{\nu_6^2+\nu_7^2}\sqrt{r_6^2+r_7^2}\\
    \leq&\,\mu(\sqrt{r_1^2+r_2^2}+\sqrt{r_4^2+r_5^2}+\sqrt{r_6^2+r_7^2})\\
    =&\,\frac{3}{2}\mu\,C_{l_1}(\rho),
\end{aligned}
\end{equation}
where $\mu=\max\{\sqrt{\nu_1^2+\nu_2^2},\sqrt{\nu_4^2+\nu_5^2},\sqrt{\nu_6^2+\nu_7^2}\}$. The diagonal terms get more in high dimensional conditions. We consider applying Cauchy-Schwarz inequality once more and obtain
\begin{equation}
\begin{aligned}
    &\nu_3 r_3+\nu_8 r_8\\
    \leq&\sqrt{\nu_3^2+\nu_8^2}\sqrt{r_3^2+r_8^2}\\
    \leq&\sqrt{\nu_3^2+\nu_8^2}\sqrt{3-(r_1^2+r_2^2+r_4^2+r_5^2+r_6^2+r_7^2)}\\
    \leq&\sqrt{\nu_3^2+\nu_8^2}\sqrt{3-\frac{(\sqrt{r_1^2+r_2^2}+\sqrt{r_4^2+r_5^2}+\sqrt{r_6^2+r_7^2})^2}{3}}\\
    =&\sqrt{\nu_3^2+\nu_8^2}\sqrt{3-\frac{3(C_{l_1}(\rho))^2}{4}}.
\end{aligned}
\end{equation}
And then we obtain an inequality similar to Eq.~\eqref{result1}
\begin{equation}\label{result2}
\begin{aligned}
&a\left(1+\mu\,C_{l_1}(\rho)+\frac{2}{3}\sqrt{\nu_3^2+\nu_8^2}\sqrt{3-\frac{3(C_{l_1}(\rho))^2}{4}}\right)\\
&\geq a(1+\frac{2}{3}\vec\nu\cdot\vec{r})=m.
\end{aligned}
\end{equation}

In a general $d$-dimensional case, the discussion can be naturally generalised, as
\begin{equation}
\begin{aligned}
    &\sum_{\substack{i=1\\i\neq j^2-1,2\leq j\leq d}}^{d^2-1}\nu_i r_i
    \leq\,\frac{d}{2}\mu\,C_{l_1}(\rho),\\
    &\mu = \max_{1\leq j<i<d}\left\{\sqrt{\nu_{(i-1)^2+2(j-1)}^2+\nu_{(i-1)^2+2j-1}^2}\right\},
\end{aligned}
\end{equation}
\begin{equation}
\begin{aligned}
    \sum_{\substack{i=1\\i= j^2-1,2\leq j\leq d}}^{d^2-1}\nu_i r_i
    \leq\sqrt{\sum_{k=2}^{d}{\nu_{k^2-1}}}\sqrt{\frac{d(d-1)}{2}-\frac{d(C_{l_1}(\rho))^2}{2(d-1)}}.
\end{aligned}
\end{equation}
Therefore we come to the result Eq.~\eqref{result3}.
\bibliographystyle{apsrev4-1}

\bibliography{bibCo}


\end{document}